\documentclass{article}
\usepackage{graphicx}
\usepackage{amsfonts}    
\usepackage{amssymb}     
\usepackage{amsmath}     
\usepackage{color}
\usepackage{amsthm}

\newtheorem{definition}{Definition}
\newtheorem{proposition}{Proposition}

\newtheorem{example}{Example}
\newtheorem{algorithm}{Algorithm}

\newcommand{\rrbr}{]\!]}
\newcommand{\llbr}{[\![}

\begin{document}

\title{A bounded operator approach to technical indicators without lag}

\author{Fr\'{e}d\'{e}ric BUTIN\footnote{Universit\'e de Lyon, Universit\'e Lyon 1, CNRS, UMR5208, Institut Camille Jordan, 43 blvd du 11 novembre 1918, F-69622 Villeurbanne-Cedex, France, butin@math.univ-lyon1.fr}}

\maketitle

\begin{abstract}
\noindent In the framework of technical analysis for algorithmic trading we use a linear algebra approach in order to define classical technical indicators as bounded operators of the space $l^\infty(\mathbb{N})$. This more abstract view enables us to define in a very simple way the no-lag versions of these tools.\\
Then we apply our results to a basic trading system in order to compare the classical Elder's impulse system with its no-lag version and the so-called Nyquist-Elder's impulse system.\\

\noindent\textbf{Keywords:} bounded operators, technical indicators without lag, linear algebra, algorithmic trading.\\

\noindent\textbf{JEL Classification:} C60, C63, G11, G15, G17.
\end{abstract}

\section{Moving averages as bounded operators}\label{firstsection}

\subsection{\textbf{Aims and organization of the paper}}

\noindent Delay in response is a major drawback of many classical technical indicators used in algorithmic trading, and this often leads to a useless or wrong information. The aim of this paper is to define in terms of bounded operators some of these classical indicators, in order to give and study their no-lag versions: these no-lag versions provide a better information that is closer to the instantaneous values of the securities, thus a better return rate of the trading system in which they occur.\\

\noindent For this purpose, we will define moving averages and exponential moving averages as bounded operators in \textbf{Section~\ref{firstsection}}. Then, \textbf{Section~\ref{secondsection}} will be devoted to the definition and the properties of the lag (in particular Proposition~\ref{lag2}); we will also make use of Nyquist criterium. Finally, all these results will enable us to give no-lag versions of famous indicators that we will compare in \textbf{Section~\ref{section3}}.\\

\noindent We denote by $E=l^\infty(\mathbb{N})$ the vector space of bounded real sequences, endowed with the norm $\|\cdot\|_\infty$ defined by $\|x\|_\infty=\sup_{n\in\mathbb{N}}|x_n|$, and by $L(E)$ the algebra of continuous endomorphisms of $E$. Then, every sequence of values\footnote{For example daily or monthly values.} of a security can be identified with an element of $E$.

\subsection{\textbf{Weighted moving averages as bounded operators}}

\noindent Let us denote by $H$ the affine hyperplane of $\mathbb{R}^p$ with equation $\displaystyle{\sum_{j=0}^{p-1}w_j=1}$ in the canonical basis, by $K$ the standard $(p-1)-$simplex (also called standard simplex of dimension $p-1$), i.e.
$$K:=\left\{\mathbf{w}=(w_0,\dots,\,w_{p-1})\in[0,\,1]^p\ /\ \sum_{j=0}^{p-1}w_j=1\right\},$$
and by $K^*$ the subset of $K$ that consists of elements whose no coordinate is zero.

\begin{definition}\label{moymob}
Let $w=(w_0,\,w_1,\dots,\,w_{p-1})$ be in $K^*$. The \emph{weighted moving average} with $p$ periods and the weights defined by the vector $w$ is the map $M_w$ from $E$ to $E$ defined, for every $x\in E$, by $M_w(x)=y$, where
$$y_n=\left\{\begin{array}{ll}
\frac{\displaystyle{\sum_{j=p-1-n}^{p-1}w_jx_{n-p+1+j}}}{\displaystyle{\sum_{j=p-1-n}^{p-1}w_j}} & \textrm{if}\ n\in\llbr0,\,p-2\rrbr\\
\displaystyle{\sum_{j=0}^{p-1}w_jx_{n-p+1+j}} & \textrm{if}\ n\geq p-1
\end{array}\right..$$
\end{definition}

\begin{proposition}\label{moymob2}
For every $w\in K^*$, $M_w$ belongs to $L(E)$.
\end{proposition}

\begin{proof}
The map $M_w$ is clearly linear. Let us prove that this map is continuous: for every $x\in E$, let us set $M_w(x)=y$ as in Definition~\ref{moymob}. Then, for every $n\in\llbr0,\,p-2\rrbr$,
\pagebreak

$$|y_n|\leq \frac{1}{\sum_{j=p-1-n}^{p-1}w_j}\sum_{j=p-1-n}^{p-1}w_j|x_{n-p+1+j}|\leq\frac{1}{\sum_{j=p-1-n}^{p-1}w_j}\sum_{j=p-1-n}^{p-1}w_j\|x\|_\infty,$$
hence $|y_n|\leq\|x\|_\infty$; and for every $n\geq p-1$,
$$|y_n|\leq \sum_{j=0}^{p-1}w_j|x_{n-p+1+j}|\leq\sum_{j=0}^{p-1}w_j\|x\|_\infty=\|x\|_\infty,$$
thus $\|M_w(x)\|_\infty\leq\|x\|_\infty$, which proves that $M_w$ is continuous.
\end{proof}

\begin{example}
$M_w^2(x)=y^{(2)}$, with $\displaystyle{y_n^{(2)}=\sum_{i=0}^{p-1}w_i\sum_{j=0}^{p-1}w_jx_{n-2p+2+i+j}}$ $\forall\ n\geq 2p-2$.
\end{example}

\begin{example}
More generally, for every $k\in\mathbb{N}$, $M_w^k(x)=y^{(k)}$, with
$$\displaystyle{y_n^{(k)}=\sum_{i_1=0}^{p-1}\sum_{i_2=0}^{p-1}\dots\sum_{i_k=0}^{p-1} w_{i_1}w_{i_2}\dots w_{i_k}x_{n-k(p-1)+i_1+i_2+\dots+i_k}}$$
for every $n\geq k(p-1)$.
\end{example}

\begin{example}
For every polynomial $\displaystyle{P=\sum_{k=0}^da_kX^k\in\mathbb{R}[X]}$, we have \mbox{$P(M_w)(x)=y$}, with
$\displaystyle{y_n=\sum_{k=0}^{d}a_ky_n^{(k)}}$ for every $n\geq d(p-1)$.\\
According to Proposition~\ref{moymob2}, $P(M_w)$ belongs to $L(E)$.
\end{example}

\subsection{\textbf{Exponential moving averages as bounded operators}}

\noindent Let us now define exponential moving averages\footnote{Note that these exponential moving averages are not weighted moving averages.}.

\begin{definition}\label{moymobexp}
Let $\alpha\in]0,\,1[$. The \emph{exponential moving average }of parameter $\alpha$ is the map $ME_\alpha$ from $E$ to $E$ defined, for every $x\in E$, by $ME_\alpha(x)=y$, where
$$y_n=\left\{\begin{array}{ll}
x_0 & \textrm{if}\ n=0\\
\alpha x_n+(1-\alpha)y_{n-1} & \textrm{if}\ n\geq1
\end{array}\right..$$
\end{definition}

\noindent For example, when $\alpha$ is equal to $\frac{2}{p+1}$, where $p\in\mathbb{N}^*$, then $p$ is called the \emph{number of periods} of $ME_\alpha$. In that case, we will denote $ME_\alpha$ by $ME_p$ instead of $ME_\frac{2}{p+1}$.\\

\noindent From Definition~\ref{moymobexp}, we immediately deduce the following proposition.

\begin{proposition}\label{moymobexp2}
Let $\alpha\in]0,\,1[$. Let $x\in E$, and set $y=ME_\alpha(x)$. Then for every $n\in\mathbb{N}$, we have $\displaystyle{y_n=(1-\alpha)^nx_0+\alpha\sum_{j=1}^n(1-\alpha)^{n-j}x_j}$.
\end{proposition}

\noindent In the same way as for weighted moving averages, we have the following proposition.

\begin{proposition}\label{moymobexp3}
For every $\alpha\in]0,\,1[$, $ME_\alpha$ belongs to $L(E)$.
\end{proposition}

\begin{proof}
The map $ME_\alpha$ is clearly linear. Let us prove that it is continuous: for every $x\in E$, let us set $ME_\alpha(x)=y$ as in Definition~\ref{moymobexp}. Then, for each $n\in\mathbb{N}$,
$$|y_n|\leq (1-\alpha)^n|x_0|+\alpha\sum_{j=1}^{n}(1-\alpha)^{n-j}|x_j|\leq\|x\|_\infty\left((1-\alpha)^n+\sum_{j=1}^{n}(1-\alpha)^{n-j}\right),$$
hence
$$|y_n|\leq \|x\|_\infty\left((1-\alpha)^n+\alpha\frac{1-(1-\alpha)^n}{\alpha}\right)=\|x\|_\infty,$$
thus $\|ME_\alpha(x)\|_\infty\leq\|x\|_\infty$, which proves that $ME_\alpha$ is continuous.
\end{proof}

\section{Lag of a weighted moving average}\label{secondsection}

\subsection{\textbf{Definition and fundamental property}}

\noindent Here, we use some results obtained by Patrick Mulloy in the article~\cite{M94} (see also \cite{B17} and \cite{E01}).\\

\noindent We denote by $\tau$ the time difference between two measures: if $x_n$ is the value at the time $t_n$, then $\tau=t_n-t_{n-1}$.

\begin{definition}\label{lag2}
Let $w=(w_0,\,w_1,\dots,\,w_{p-1})$ be in $\mathbb{R}^p$. Let $M$ be a linear map from $E$ to $E$. For every $x\in E$, let us set $y=M(x)$. We assume that $M$ satisfies
$$\forall\ n\geq p-1,\ y_n=\displaystyle{\sum_{j=0}^{p-1}w_jx_{n-p+1+j}}.$$
Then the \emph{lag} of $M$ is defined by $lag(M)=\displaystyle{\tau\sum_{j=0}^{p-1}w_j(p-1-j)}$.\\
\end{definition}

\noindent Let us note that $lag(id)=0$. Let us also note that this definition is of course valid for every weighted moving average (i.e.~with $w\in K^*$). However, it is valid for more general linear maps.

\begin{example}\label{ax1}
For $w=\frac{1}{p}(1,\,1,\dots,\,1)\in K^*$, the lag of $M_w$ (\emph{classical moving average} with $p$ periods) is $lag(M)=\displaystyle{\frac{\tau}{p}\sum_{j=0}^{p-1}(p-1-j)=\frac{(p-1)\tau}{2}}$.\\
\end{example}

\begin{example}\label{ax2}
For $w=\frac{2}{p(p+1)}(1,\,2,\dots,\,p)\in K^*$, the lag of $M_w$ (\emph{simple weighted moving average} with $p$ periods) is given by $$lag(M)=\displaystyle{\frac{2\tau}{p(p+1)}\sum_{j=0}^{p-1}(j+1)(p-1-j)=
\frac{2\tau}{p(p+1)}\sum_{j=1}^{p}j(p-j)=\frac{(p-1)\tau}{3}}.$$
\end{example}

\noindent The following proposition establishes the fundamental property of the lag.

\begin{proposition}\label{lag2}
For every polynomial $\displaystyle{P=\sum_{k=0}^da_kX^k\in\mathbb{R}[X]}$, we have the formula $\displaystyle{lag(P(M_w))=lag(M_w)\sum_{k=0}^dka_k=lag(M_w)P'(1)}$.
\end{proposition}

\begin{proof}
For every $k\in\mathbb{N}$, we have
$$\begin{array}{lll}
lag(M_w^k) & = & \displaystyle{\tau\sum_{i_1=0}^{p-1}\sum_{i_2=0}^{p-1}\dots\sum_{i_k=0}^{p-1} w_{i_1}w_{i_2}\dots w_{i_k}(k(p-1)-i_1-i_2-\dots-i_k)}\\
     & = & \displaystyle{\sum_{l=1}^{k}\tau\sum_{i_1=0}^{p-1}\sum_{i_2=0}^{p-1}\dots\sum_{i_k=0}^{p-1} w_{i_1}w_{i_2}\dots w_{i_k}(p-1-i_l)}\\
     & = & \displaystyle{\sum_{l=1}^{k}\tau\sum_{i_l=0}^{p-1} w_{i_l}(p-1-i_l)}
  \end{array}$$
since $\displaystyle{\sum_{j=0}^{p-1}w_j=1}$. Hence $\displaystyle{lag(M_w^k)=\sum_{l=1}^{k}lag(M_w)=k\,lag(M_w)}$.\\
Finally, $\displaystyle{lag(P(M_w))=\sum_{k=0}^da_k\,lag(M_w^k)=lag(M_w)\sum_{k=0}^dka_k}$.
\end{proof}

\subsection{\textbf{Weighted moving averages without lag}}

\noindent We can now get no-lag versions of weighted moving averages: this is the aim of Propositions~\ref{lag3} and~\ref{lag3b}.

\begin{proposition}\label{lag3}
The only polynomial of the form $P=aX+bX^2$ such that $a+b=1$ and $lag(P(M_w))=0$ is $P=2X-X^2$.\\
\end{proposition}

\begin{proof}
We have $a+b=1$. And according to Proposition~\ref{lag2}, $lag(P(M_w))=(a+2b)lag(M_w)$, thus $a+2b=0$, so that the unique solution of the system is $(a,\,b)=(2,\,-1)$.
\end{proof}

\begin{proposition}\label{lag3b}
The only polynomial of the form $Q=aX+bX^2+X^3$ such that $a+b=0$ and $lag(Q(M_w))=0$ is $Q=3X-3X^2+X^3$.
\end{proposition}

\begin{proof}
In the same way as in Proposition~\ref{lag3}, we have $a+b=0$ and $lag(Q(M_w))=(a+2b+3)lag(M_w)$, thus $a+2b+3=0$, and the unique solution of the system is $(a,\,b)=(3,\,-3)$.
\end{proof}

\subsection{\textbf{Using Nyquist criterium}}

\noindent Here we recall some results obtained by D\"{u}rschner in his article~\cite{D12}. We consider two weighted moving averages $M_{w_1}$ and $M_{w_2}$ with respectively $p_1$ and $p_2$ periods. Let $x\in E$, and let us set $y^{(1)}=M_{w_1}(x)$
and $y^{(2)}=M_{w_2}(y^{(1)})$. We consider the angles $a_1$ and $a_2$ as on Figure~\ref{fig:1}.

\begin{figure}[h]
\includegraphics[height=6cm]{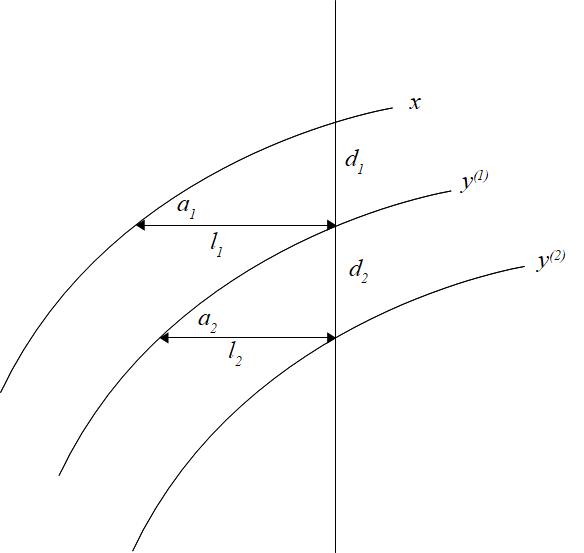}
\caption{Nyquist criterium}
\label{fig:1}
\end{figure}

\noindent We have $a_1\simeq a_2$, hence $\frac{d_1}{l_1}\simeq\frac{d_2}{l_2}$, i.e. $\frac{l_1}{l_2}\simeq\frac{d_1}{d_2}$. Moreover, $d_1=x_n-y_n^{(1)}$ and $d_2=y_n^{(1)}-y_n^{(2)}$,
thus $\alpha_n:=\frac{l_1}{l_2}\simeq\frac{x_n-y_n^{(1)}}{y_n^{(1)}-y_n^{(2)}}$.\\
We now define $z=(z_n)_{n\in\mathbb{N}}$ such that for every $n\in\mathbb{N}$, $\alpha_n:=\frac{l_1}{l_2}=\frac{z_n-y_n^{(1)}}{y_n^{(1)}-y_n^{(2)}}$, i.e. $z_n=(1+\alpha_n)y_n^{(1)}-\alpha_n y_n^{(2)}$.\\
Now, we want $\alpha_n$ not to depend on $n$. For this, we consider the lags of $M_{w_1}$ and $M_{w_2}$ that equal $\frac{(p_1-1)\tau}{2}$ and $\frac{(p_2-1)\tau}{2}$ when $M_{w_1}$ and $M_{w_2}$ are classical moving averages, so that $\frac{lag(M_{w_1})}{lag(M_{w_2})}=\frac{p_1-1}{p_2-1}$, and we set $\alpha=\frac{p_1-1}{p_2-1}$. Finally, we define a new weighted moving average $N$ by $z_n=(1+\alpha)y_n^{(1)}-\alpha y_n^{(2)}$ for every $n\in\mathbb{N}$, where $z=N(x)$. In terms of linear maps, we have $$N=(1+\alpha)M_{w_1}-\alpha M_{w_2}\circ M_{w_1}.$$
The stability Nyquist criterium concerning the choice of $p_1$ and $p_2$ is $\displaystyle{\frac{p_1}{p_2}\geq 2}$.

\begin{definition}\label{defMMPNy}
Let $M_{w_1}$ and $M_{w_2}$ be two \emph{simple} weighted moving averages\footnote{According to example~\ref{ax2}, their weights are defined by $w_1=\frac{2}{p_1(p_1+1)}(1,\,2,\dots,\,p_1)$ and $w_2=\frac{2}{p_2(p_2+1)}(1,\,2,\dots,\,p_2)$.} with respectively $p_1$ and $p_2$ periods and $p_1\geq 2\,p_2$. Let us set $\alpha=\frac{p_1-1}{p_2-1}$.\\
Then, we call \emph{Nyquist moving average} with periods $p_1,\,p_2$ the element $N_{p_1,p_2}$ of $L(E)$ defined by $N_{p_1,p_2}=(1+\alpha)M_{w_1}-\alpha M_{w_2}\circ M_{w_1}$.
\end{definition}

\section{Technical indicators without lag}\label{section3}

\noindent Here, we make use of the results of Section~\ref{secondsection} in order to define technical indicators ``without lag''. We begin with the exponential moving average, then define the MACD and use them for Elder's impulse system.

\subsection{\textbf{Exponential moving average without lag}}

\noindent For every $\alpha\in]0,\,1[$, we \emph{define} the exponential moving average without lag as the element $ME_{\alpha,wl}:=P(ME_\alpha)$ of $L(E)$, where $P$ is defined in Proposition~\ref{lag3}.

\subsection{\textbf{MACD without lag}}

\noindent $\bullet$ Let us recall the definition of the MACD (moving average convergence divergence) introduced by Gerald Appel in 1979 in his financial newsletter ``Systems and Forecasts'' and presented in his book~\cite{A85}.\\
First we set $MACD=ME_{12}-ME_{26}$, then $MACDS=ME_9\circ MACD$, and $MACDH=MACD-MACDS$. These three maps belong to $L(E)$.\\

\noindent $\bullet$ We \emph{define} the MACD without lag as follows: first we set $$MACD_{wl}=ME_{12,wl}-ME_{26,wl},$$ then $$MACDS_{wl}=ME_{9,wl}\circ MACD_{wl},$$ and $$MACDH_{wl}=MACD_{wl}-MACDS_{wl}.$$ These three maps also belong to $L(E)$.\\

\noindent $\bullet$ In the same way, we \emph{define} the Nyquist-MACD as follows: first we set $$MACD_{N}=N_{12,3}-N_{26,6},$$ then $$MACDS_{N}=N_{9,3}\circ MACD_{N},$$ and $$MACDH_{N}=MACD_{N}-MACDS_{N}.$$ Here again these three maps belong to $L(E)$.

\subsection{\textbf{Elder's impulse system without lag}}\label{parElder}

\noindent By making use of the exponential moving average without lag and the MACD without lag, we can give a ``no-lag version'' of Elder's impulse system. Let us first recall the definition of the impulse system introduced by Alexander Elder in his best-seller~\cite{E02}.\\
Let us set $C=\{R,\,G,\,B\}$\footnote{``R'' (resp. ``G'', ``B'') stands for ``red'' (resp. ``green'', ``blue'').} and denote by $F$ the set $C^\mathbb{N}$ of sequences with values in $C$. Then, Elder's impulse system can be defined in an algorithmic way as follows.

\begin{definition}\label{IS}
Elder's impulse system is the map $IS$ from $E$ to $F$ defined, for every $x\in E$, by $IS(x)=y$, where $y_0=B$ and for every $n\in\mathbb{N}^*$,
$$y_n=\left\{\begin{array}{ll}
                 G & \textrm{if}\ ME_{12}(x)_n>ME_{12}(x)_{n-1}\ \textrm{and}\ MACDH(x)_n>MACDH(x)_{n-1}\\
                 R & \textrm{if}\ ME_{12}(x)_n<ME_{12}(x)_{n-1}\ \textrm{and}\ MACDH(x)_n<MACDH(x)_{n-1} \\
                 B & \textrm{else}
               \end{array}
\right..$$
\end{definition}

\noindent Let us now give a ``no-lag version'' of Elder's impulse system

\begin{definition}\label{ISSR}
Elder's impulse system without lag is the map $IS_{wl}$ from $E$ to~$F$ defined, for every $x\in E$, by $IS_{wl}(x)=y$, where $y_0=B$ and for every $n\in\mathbb{N}^*$,
\begin{small}
$$y_n=\left\{\begin{array}{ll}
                 G & \textrm{if}\ ME_{12,wl}(x)_n>ME_{12,wl}(x)_{n-1}\ \textrm{and}\ MACDH_{wl}(x)_n>MACDH_{wl}(x)_{n-1}\\
                 R & \textrm{if}\ ME_{12,wl}(x)_n<ME_{12,wl}(x)_{n-1}\ \textrm{and}\ MACDH_{wl}(x)_n<MACDH_{wl}(x)_{n-1} \\
                 B & \textrm{else}
               \end{array}
\right..$$
\end{small}
\end{definition}

\noindent We also define the Nyquist-Elder's impulse system.

\begin{definition}\label{ISSRN}
Nyquist-Elder's impulse system is the map $IS_{N}$ from $E$ to~$F$ defined, for every $x\in E$, by $IS_{N}(x)=y$, where $y_0=B$ and for every $n\in\mathbb{N}^*$,
\begin{small}
$$y_n=\left\{\begin{array}{ll}
                 G & \textrm{if}\ N_{12,3}(x)_n>N_{12,3}(x)_{n-1}\ \textrm{and}\ MACDH_{N}(x)_n>MACDH_{N}(x)_{n-1}\\
                 R & \textrm{if}\ N_{12,3}(x)_n<N_{12,3}(x)_{n-1}\ \textrm{and}\ MACDH_{N}(x)_n<MACDH_{N}(x)_{n-1} \\
                 B & \textrm{else}
               \end{array}
\right..$$
\end{small}
\end{definition}

\subsection{\textbf{Comparison of the three versions of Elder's impulse system}}

\noindent Here we use a very simple trading system\footnote{In this trading system, every position is automatically closed one day before the end of the test.} in order to compare the three versions of Elder's impulse system given in section~\ref{parElder}. See for example \cite{K19} and \cite{JT09} for more information about trading systems. We use $x=(x_n)_{n\in\llbr0,\,d\rrbr}$ the daily values of S\&P index on the time period from 2017-11-01 to 2018-10-31.

\begin{algorithm}\label{trsyst}
(very simple trading system)\\
\verb"For n" $\in\llbr0,\,d\rrbr$\verb":"\\
\verb"    Long entry : if f(x)_n=G"\\
\verb"        then buy 1 mini contract"\\
\verb"    Short entry : if f(x)_n=R"\\
\verb"        then sell short 1 mini contract"\\
\verb"    Long exit : if f(x)_n=R"\\
\verb"        then sell 1 mini contract"\\
\verb"    Short exit : if f(x)_n=G"\\
\verb"        then buy 1 mini contract"
\end{algorithm}

\noindent We use Algorithm~\ref{trsyst} with $f=IS$, $f=IS_{wl}$ and $f=IS_N$. The results\footnote{Here the transaction cost for every entry/exit is $\$3$.} are given by Table~\ref{tab:1}, in which all values are in USD. Let us note that the value of the S\&P index is $2\,572.625$ (resp. $2\,706.125$) on 2017-11-01 (resp. 2018-10-31), hence a profit of $TPI=\$ 6\,675$ for one mini contract during this period.

\begin{table}[h]
\caption{Comparison of the three various trading systems}
\label{tab:1}
\begin{tabular}{||l||r|r|r||}
\hline\noalign{\smallskip}
 & $f=IS$ & $f=IS_{wl}$ & $f=IS_N$\\
\noalign{\smallskip}\hline\noalign{\smallskip}
\textrm{Number of trades} & $27$ & $35$ & $52$ \\
\textrm{Total net profit} & $13\,549$ & $20\,489$ & $31\,395$ \\
\textrm{Percentage of winning trades} & $52\%$ & $46\%$ & $50\%$\\
\textrm{Average net profit per trade} & $502$ & $585$ & $604$ \\
\noalign{\smallskip}\hline\noalign{\smallskip}
\textrm{Total net profit of winning trades ($TP$)} & $32\,452$ & $41\,278$ & $49\,835$ \\
\textrm{Average net profit per winning trade ($AP$)} & $2\,318$ & $2\,580$ & $1\,917$ \\
\noalign{\smallskip}\hline\noalign{\smallskip}
\textrm{Total net lost of losing trades ($TL$)} & $-18\,903$ & $-20\,789$ & $-18\,440$ \\
\textrm{Average net lost per losing trade ($AL$)} & $-1\,454$ & $-1\,094$ & $-709$ \\
\textrm{Greatest lost between two winning trades} & $-10\,231$ & $-7\,793$ & $-3\,761$ \\
\noalign{\smallskip}\hline\noalign{\smallskip}
\textrm{Total net profit of long trades} & $11\,322$ & $10\,249$ & $14\,930$ \\
\textrm{Average net profit per long trade} & $755$ & $539$ & $515$ \\
\noalign{\smallskip}\hline\noalign{\smallskip}
\textrm{Total net profit of short trades} & $2\,227$ & $10\,240$ & $16\,465$ \\
\textrm{Average net profit per short trade} & $186$ & $640$ & $716$ \\
\noalign{\smallskip}\hline\noalign{\smallskip}
\textrm{Profit factor $TP/|TL|$} & $1.72$ & $1.99$ & $2.70$ \\
\textrm{Ratio $AP/|AL|$} & $1.59$ & $2.36$ & $2.70$ \\
\textrm{Ratio $TP/TPI$} & $4.86$ & $6.18$ & $7.46$ \\
\noalign{\smallskip}\hline
\end{tabular}
\end{table}

\noindent We observe that \emph{here} the Nyquist-Elder's impulse system is much better than the Elder's impulse system without lag, which is itself better than the classical impulse system: the information given by Nyquist-Elder's impulse system is indeed closer to the instantaneous value of the index since it has less delay than the classical impulse system. We can also note that the number of trades as well as the average net profit per trade are increasing when $f$ is respectively equal to $IS$, $IS_{wl}$ and $IS_N$. And the repartition of profit among long and short trades is more uniform with $IS_{wl}$ and $IS_N$ than with $IS$.\\

\noindent Figure~\ref{fig:2} eventually shows the S\&P index from 2017-11-01 to 2018-10-31, with the Nyquist moving averages $N_{12,3}$ (dotted line) and $N_{26,3}$ (solid line), the graphs of $MACD_N$ and $MACDS_N$ with the histogram $MACDH_N$, and the three versions of Elder's impulse system (from bottom to top: $IS$, $IS_{wl}$ and $IS_N$).

\begin{figure}[h]
\includegraphics[height=6cm]{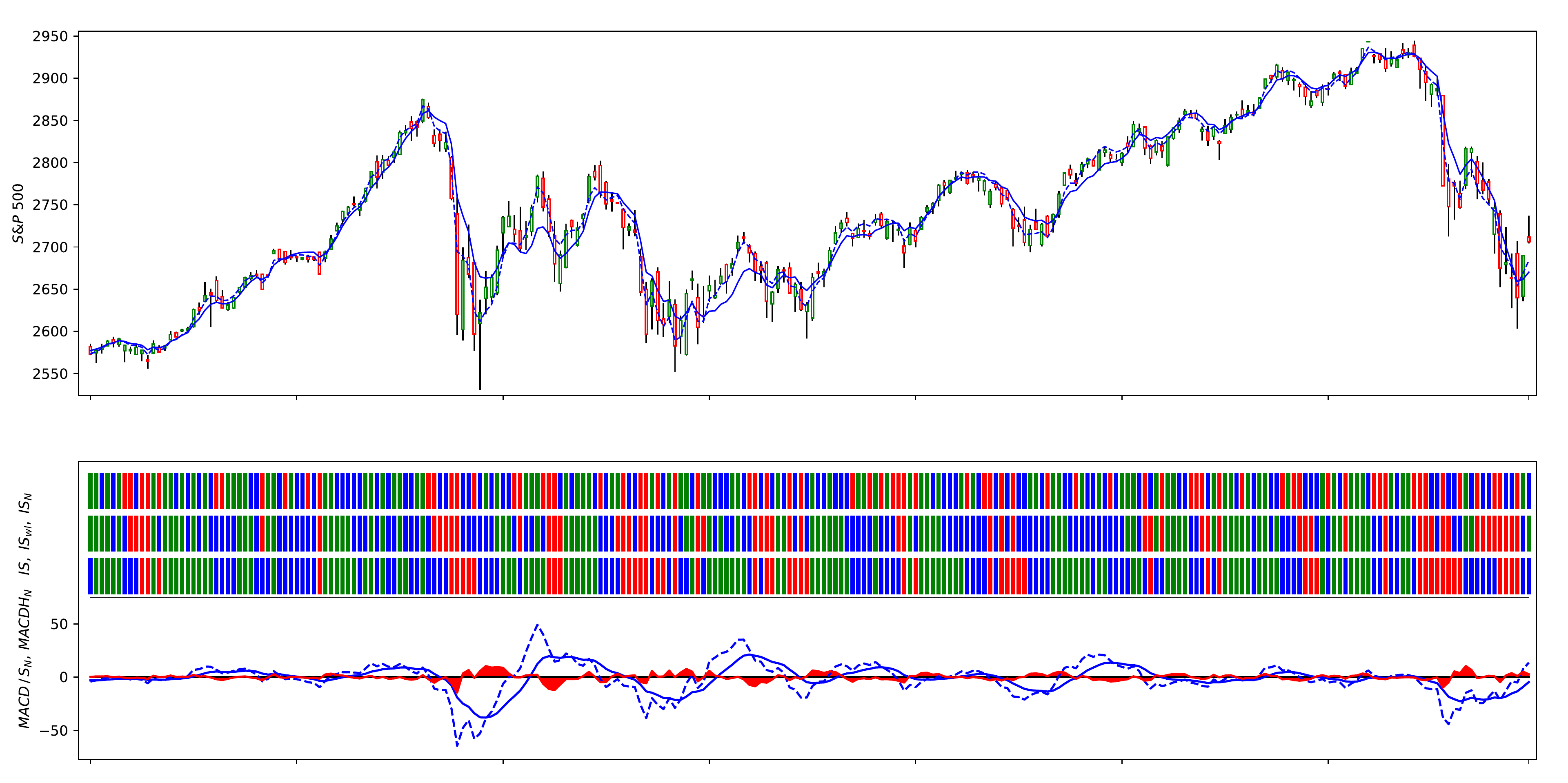}
\caption{S\&P index from 2017-11-01 to 2018-10-31}
\label{fig:2}
\end{figure}

\vspace{-.5cm}

\begin{small}

\end{small}

\end{document}